\documentclass[12pt,letterpaper]{article}

\usepackage{etoolbox}
\newtoggle{REV}\toggletrue{REV}
\togglefalse{REV} 

\newtoggle{SUPPLEMENTAL}\toggletrue{SUPPLEMENTAL}
\togglefalse{SUPPLEMENTAL} 

\newtoggle{BLINDED}\toggletrue{BLINDED}
\togglefalse{BLINDED} 

\usepackage{graphicx,grffile}
\usepackage{amsmath,amssymb,amsthm}
\usepackage{mathtools,dsfont,centernot}
\usepackage{caption}
\usepackage{subcaption}
\usepackage{booktabs,tabularx,threeparttable,longtable,ragged2e}
\usepackage{siunitx}
\usepackage[shortlabels]{enumitem}
\usepackage{alltt}
\usepackage[longnamesfirst]{natbib}
\usepackage{hypernat}
\usepackage[nodisplayskipstretch]{setspace}
\usepackage[bottom]{footmisc}
\usepackage{xr} 

\usepackage[final]{listings}
\lstdefinestyle{inlineR}{language=R,frame=none,basicstyle=\ttfamily,keywordstyle=\ttfamily,stringstyle=\ttfamily,keepspaces=true,showspaces=false,showstringspaces=false,breaklines=true,upquote=true,print,columns=fullflexible}
\newcommand{\code}{\lstinline}

\usepackage[hyphens]{url}
\usepackage{hyperref}

\usepackage[margin=1in,letterpaper]{geometry}
\usepackage[capitalize,noabbrev]{cleveref}

  \renewenvironment{thebibliography}[1]%
  {\begin{oldthebibliography}{#1}\setlength{\parskip}{0ex}\setlength{\itemsep}{0ex}}%
  {\end{oldthebibliography}}
\appto\TPTnoteSettings{\linespread{1}\footnotesize}

\DeclareGraphicsExtensions{.pdf,.png}
\graphicspath{ {./figures/} {../figures/} }

\newcommand{\citeposs}[1]{\citeauthor{#1}'s (\citeyear{#1})}

\crefname{conjecture}{Conjecture}{Conjectures}
\crefname{section}{Section}{Sections}
\crefname{subsection}{Section}{Sections}
\crefname{subsubsection}{Section}{Sections}
\Crefname{conjecture}{Conjecture}{Conjectures}
\Crefname{section}{Section}{Sections}
\Crefname{subsection}{Section}{Sections}
\Crefname{subsubsection}{Section}{Sections}
\crefname{appendix}{Appendix}{Appendices}
\crefname{subappendix}{Appendix}{Appendices}
\crefname{subsubappendix}{Appendix}{Appendices}
\Crefname{appendix}{Appendix}{Appendices}
\Crefname{subappendix}{Appendix}{Appendices}
\Crefname{subsubappendix}{Appendix}{Appendices}
\crefname{equation}{}{}
\Crefname{equation}{Equation}{Equations}
\crefname{enumi}{}{}
\Crefname{enumi}{}{}

\crefname{assumption}{}{}
\Crefname{assumption}{Assumption}{Assumptions}

\Crefname{method}{Method}{Methods}
\newlist{steps}{enumerate}{1}
\setlist[steps]{label=\arabic*., ref=\arabic*, itemsep=0pt}
\crefname{stepsi}{Step}{Steps}
\Crefname{stepsi}{Step}{Steps}

\theoremstyle{plain}
\newtheorem{theorem}{Theorem}
\newtheorem{lemma}[theorem]{Lemma}

\newtheorem{corollary}[theorem]{Corollary}

\theoremstyle{definition}

\usepackage{bm} 

\newcommand{\iid}{\overset{\mathit{iid}}{\sim}}
\newcommand{\deq}{\overset{d}{=}}

\newcommand{\dconv}{\xrightarrow{d}}

\newcommand{\MSE}{\textrm{MSE}}
\newcommand{\CP}{\textrm{CP}}

\DeclareMathOperator{\Corr}{Corr}

\newcommand{\R}{{\mathbb R}}

\DeclareMathOperator{\E}{E} 

\let\Pr\relax \DeclareMathOperator{\Pr}{P} 

\DeclareMathOperator*{\argmin}{arg\,min}

\newcommand{\CI}{\textrm{CI}}
\newcommand{\NormDist}{\mathrm{N}}

\newcommand{\diff}{\mathop{}\!\mathrm{d}} 

\newcommand{\independenT}[2]{\mathrel{\rlap{$#1#2$}\mkern2mu{#1#2}}}
\newcommand\independent{\protect\mathpalette{\protect\independenT}{\perp}} 
\providecommand{\abs}[1]{\lvert#1\rvert}

\let\originalleft\left
\let\originalright\right
\renewcommand{\left}{\mathopen{}\mathclose\bgroup\originalleft}
\renewcommand{\right}{\aftergroup\egroup\originalright}

\newcommand{\mockalph}[1]{}  
\allowdisplaybreaks[3]

\hypersetup{
  pdfauthor = {David M.\ Kaplan and Xin Liu},
  pdfkeywords = {economics, econometrics, statistics},
  pdftitle = {Kaplan and Liu: coverage probability with intentional bias},
  pdfsubject = {econometrics},
  pdfpagemode = UseNone
}

\title{Confidence intervals for intentionally biased estimators}
\author{%
David M.\ Kaplan\thanks{Department of Economics, University of Missouri, \texttt{kaplandm@missouri.edu}}
\and
Xin Liu\thanks{School of Economic Sciences, Washington State University, \texttt{xin.liu1@wsu.edu}}
}

\date{October 31, 2023}

\begin{document}

\parindent=1.5em

\maketitle

\doublespacing 

\begin{abstract}
We propose and study three confidence intervals (CIs) centered at an estimator that is intentionally biased to reduce mean squared error.
The first CI simply uses an unbiased estimator's standard error; compared to centering at the unbiased estimator, this CI has higher coverage probability for confidence levels above $91.7\%$, even if the biased and unbiased estimators have equal mean squared error.
The second CI trades some of this ``excess'' coverage for shorter length.
The third CI is centered at a convex combination of the two estimators to further reduce length.
Practically, these CIs apply broadly and are simple to compute.




\textit{Disclosure statement}: we (the authors) have no competing interests to declare.
\end{abstract}

\clearpage

\newcommand{\paperspacing}{\onehalfspacing}
\renewcommand{\paperspacing}{\doublespacing}
\paperspacing

\section{Introduction}
\label{sec:intro}


We propose simple confidence intervals based on estimators that intentionally introduce bias in order to reduce mean squared error (MSE).
Such estimators have a long history and continued popularity.
For example, the classic averaging/shrinkage approach of \citet{JamesStein1961} was applied to simultaneous equation models by \citet{Maasoumi1978}, which in turn inspired the recent asymptotic risk dominance results of \citet{Hansen2017} for a closely related 2SLS/OLS averaging estimator, while others have considered Stein-like GMM averaging with misspecified moments \citep{ChengLiaoShi2019,DiTraglia2016,Liu2022}.
As another example, the $L_2$ penalty of ridge regression \citep{HoerlKennard1970} has been modified to get lasso \citep{Tibshirani1996}, bridge \citep{KnightFu2000}, and SCAD \citep{FanLi2001}, with additional optimality results by \citet{Zou2006} and \citet{ChetverikovEtAl2021}, among others.
A third category includes smoothed estimators.
Even aside from nonparametric estimators that increase bias in order to minimize MSE, examples include \citeposs{Horowitz1992} smoothed version of the binary choice maximum score estimator \citep{Manski1975} and \citeposs{GroeneboomEtAl2010} smoothed version of the nonparametric maximum likelihood estimator of the current status model \citep{GroeneboomWellner1992}.
Another example is \citeposs{KaplanSun2017} smoothed estimator of instrumental variables quantile regression \citep{ChernozhukovHansen2005,ChernozhukovHansen2006} and quantile regression \citep{KoenkerBassett1978}; even stronger theoretical results for smoothed quantile regression are developed by \citet{FernandesEtAl2021} and \citet{HeEtAl2023}, among others.
As detailed in \cref{sec:setup}, although our results do not apply universally to all these examples, they apply to several, with smoothed quantile regression as the example we detail throughout the paper.


In finite samples, ignoring the bias is problematic for the default confidence interval (CI) centered at the biased estimator and using its own standard error.
Such a CI always ``undercovers,'' meaning coverage probability is below the desired confidence level.
Sometimes asymptotic arguments are made that the bias goes to zero fast enough that the coverage probability increases toward $1-\alpha$ in the limit, but this does not fully address the finite-sample reality.


In this paper, we propose two CIs that are centered at the intentionally biased (lower MSE) estimator but attain at least the desired coverage probability, and a third CI centered at a convex combination of the unbiased and biased estimators.
The first CI uses the unbiased estimator's standard error, while the second additionally uses the biased estimator's standard error in order to trade some ``excess'' coverage probability for shorter length.
Both CIs have advantages over the other common approaches at the $95\%$ confidence level.
Compared to using the biased estimator's CI mentioned above, our CIs do not suffer from undercoverage.
Compared to CIs based on bias-correction, our CIs do not require any model-specific knowledge or bias estimation.
Bias-corrected CIs work well in certain settings, but generally the bias can depend on parameters that are very challenging to estimate, like high-order derivatives of conditional densities of unobserved error terms as in (10) of \citet{KaplanSun2017} for smoothed instrumental variables quantile regression and as in Theorem 1 of \citet{FernandesEtAl2021} for smoothed quantile regression, and properly accounting for the extra uncertainty from bias estimation can make CIs longer than necessary, as noted by \citet[p.~2]{ArmstrongKolesar2020} in the context of regression discontinuity.
That is, the relative practical convenience and statistical properties of our CIs compared to bias-corrected CIs depends on the specific model and setting.
Compared to centering at the unbiased estimator, both our CIs have higher coverage probability, and our second CI has shorter length.
By using the correlation between the unbiased and intentionally biased estimators, CI length can be further reduced by basing the CI on a convex combination of the two estimators.


Our first results focus on the coverage probability (CP) of our first CI compared with the benchmark CI that is centered at the unbiased estimator in addition to using its standard error.
Our CI has the same length by construction, yet at the conventional $95\%$ confidence level, our CI has higher CP than the benchmark, even if the estimators' MSEs are equal (so the biased estimator is not actually better).
Higher CP means fewer coverage errors (CI fails to contain true value).
Although typically the goal is to minimize length subject to an upper bound on the coverage error rate, given a fixed CI length we prefer a lower error rate, or equivalently a higher CP.
%
More generally, given equal MSE, our CP is higher than the benchmark whenever the confidence level is at least $91.7\%$, regardless of the magnitude of bias, and our worst-case CP is $90.00\%$ (rounded) at the $90\%$ confidence level, too.
When the biased estimator has lower MSE, then our CP is even higher.
However, at lower confidence levels, our CI's undercoverage can be significant, and CP can be near zero for confidence levels below $68.3\%$ when the bias is large.
That is, our CI's unambiguous advantage over the benchmark CI at conventional confidence levels is not trivial because it disappears at lower confidence levels.

The results for our second CI show how it can further achieve shorter length than the benchmark.
Essentially, it calibrates the critical value to achieve exact coverage probability if the two estimators have the same MSE, given the ratio of the two estimators' standard errors.
Equivalently, ``same MSE'' can be interpreted as using the upper bound for bias like in \citet{ArmstrongKolesar2021}, who in turn build upon the fixed-length CI results of \citet{Donoho1994}.
From our first results, at confidence level $95\%$, we know that if the biased estimator's standard error is strictly smaller than the unbiased estimator's standard error, then our first CI has CP strictly above $95\%$, even in the worst case of equal MSE.
Thus, we can shorten the CI's length and still achieve at least $95\%$ CP, and even higher if the biased estimator's MSE is strictly below the unbiased estimator's MSE.
This is what our second CI does.
Compared to our first CI, one practical downside is the additional reliance on the biased estimator's standard error, but if such an estimate is readily available, then we recommend our second CI in practice because it has correct coverage while being more precise.

If the setting is strengthened to joint normality with know correlation in addition to known variances, then this second CI can be further shortened by using a convex combination of the unbiased and biased estimators.
Our second CI is the special case with all weight on the biased estimator; by searching over all possible convex combinations, we can find the shortest possible CI.
The benefit of this is highest when the correlation is lower and the original unbiased and biased estimators' variances are similar.

The work most similar to ours is from \citet{ArmstrongKolesar2021}.
The limiting experiment in their (13), which is an asymptotic version of the setting in their Section 2, is similar to our strongest setting with joint normality and known correlation.
Indeed, our third CI based on a linear combination of the unbiased and biased estimators is inspired by their use of a linear estimator.
However, our work also differs from theirs in several ways.
First, our first CI relies only on the unbiased estimator's standard error, which is often easy and fast to compute, and even our second CI does not require knowledge of the correlation.
Second, our bias bound is implied by the intentionally biased estimator's lower MSE, rather than require separate knowledge.
Third, our setting always has an unbiased estimator (and thus valid CI) available; we show how at conventional confidence levels we can improve upon this, even though at other confidence levels this is not generally true.
Of course, \citet{ArmstrongKolesar2021} provide results that apply to general settings different than ours, like generalized method of moments.

After describing the setup in \cref{sec:setup}, in \cref{sec:equal-MSE} we characterize properties of our first proposed CI in the case when the biased estimator's MSE is identical to the unbiased estimator's MSE.
We consider CP as a function of bias as well as the nominal confidence level.
We characterize the cases for which our CI has higher CP than the CI centered at the unbiased estimator.
Then in \cref{sec:lower-MSE} we show that if the biased estimator has strictly lower MSE, the CP becomes even higher, specifically if we reduce the bias while keeping the variance fixed.
\Cref{sec:bias-bound} describes our second CI and establishes its properties.
\Cref{sec:lincom} extends the second CI to an optimal convex combination of the unbiased and biased estimators.
\Cref{sec:sim,sec:emp} contain simulation and empirical results.

\paragraph{Notation and abbreviations}
For notation, 
$\NormDist(\mu , \sigma^2)$ is the normal distribution with mean $\mu$ and variance $\sigma^2$, and $\Phi(\cdot)$ and $\phi(\cdot)$ are respectively the CDF and PDF of the standard normal distribution $\NormDist(0,1)$, whose $p$-quantile is denoted $z_p$; the sine, cosine, tangent, and secant functions are respectively $\sin(\cdot)$, $\cos(\cdot)$, $\tan(\cdot)$, and $\sec(\cdot)$.
Acronyms used include those for 
confidence interval (CI), 
coverage probability (CP), 
cumulative distribution function (CDF),
generalized method of moments (GMM),
mean squared error (MSE),
ordinary least squares (OLS),
smoothed quantile regression (SQR),
probability density function (PDF),
quantile regression (QR),
smoothed quantile regression (SQR),
smoothly clipped absoluted deviation (SCAD),
and
two-stage least squares (2SLS).

\section{Setup}
\label{sec:setup}

We describe the estimators and their sampling distributions in \cref{sec:setup-dist}, some possible confidence intervals in \cref{sec:setup-CIs}, and how well our setting applies to specific estimators in \cref{sec:setup-apps}.

\subsection{Estimators and distributions}
\label{sec:setup-dist}

Given a scalar parameter $\theta$, consider the two estimators
\begin{equation}
\label{eqn:theta-hat-dists}
\begin{split}
&\hat\theta_1\sim\NormDist(\theta, s_1^2), \\
&\hat\theta_2\sim\NormDist(\theta+b_2, s_2^2)
,
\end{split}
\end{equation}
where $s_1$ and $s_2$ are the respective standard errors and $b_2$ is the bias of $\hat\theta_2$.
Estimator $\hat{\theta}_2$ is intentionally biased to reduce MSE:
\begin{equation}
\label{eqn:lower-MSE}
\MSE(\hat\theta_2) = b_2^2+s_2^2
\le s_1^2 = \MSE(\hat\theta_1) .
\end{equation}

The normal distributions in \cref{eqn:theta-hat-dists} represent an asymptotic approximation.
Our \cref{eqn:theta-hat-dists} is very similar to settings used in other papers; for example, it is a weaker version of the main example setting on page 7 of \citet{ArmstrongEtAl2023}, who instead study optimal estimation and more strongly assume joint normality with known covariance, and it is similar to the limiting experiment in (13) of \citet{ArmstrongKolesar2021}.
The goal of \cref{eqn:theta-hat-dists} is simply to capture bias in a meaningful way, as opposed to asymptotic approximations in which the bias completely disappears.
For example, given some $r>0$ and $q>0$, imagine $n^r(\hat\theta_2-\theta-n^{-q}\tilde{b}_2)\dconv\NormDist(0,\tilde{s}_2^2)$, so approximately $\hat\theta_2\sim\NormDist(\theta+n^{-q}\tilde{b}_2, \tilde{s}_2^2/n^{2r})$.
Given a value of sample size $n$, our \cref{eqn:theta-hat-dists} uses $b_2=n^{-q}\tilde{b}_2$ and $s_2^2=\tilde{s}_2^2/n^{2r}$.
Some papers argue that asymptotically we can ignore the bias if $q>r$ so that the order $n^{-q}$ bias is of smaller order of magnitude than the $n^{-r}$ standard deviation, but we want capture the effect of bias that can be important in finite samples, while also allowing zero bias ($b_2=0$) as a special case.
For $\hat\theta_1$, the finite-sample bias does not need to be exactly zero as long as \cref{eqn:theta-hat-dists} provides a good approximation, meaning the bias of $\hat\theta_1$ is negligible compared to $b_2$, $s_1$, and $s_2$, like for the quantile regression estimator $\hat\theta_1$ in our simulation (\cref{sec:sim}).
If $\hat\theta_2$ is a smoothed estimator, then $\hat\theta_1$ could be a highly under-smoothed estimator, although optimal bandwidths in such cases are beyond our scope.

The scalar parameter $\theta$ can be a summary of an underlying vector-valued or function-valued parameter as long as \cref{eqn:theta-hat-dists} holds, but confidence sets or bands for non-scalar $\theta$ are beyond our scope.

In \cref{eqn:theta-hat-dists}, we treat some parameters as known and some as unknown.
Viewing \cref{eqn:theta-hat-dists} as an asymptotic approximation, ``known'' means ``can be estimated consistently.''
We always assume the unbiased estimator's standard error $s_1$ is known and that the biased estimator $\hat\theta_2$ can be computed.
The first CI we propose does not require any further knowledge.
Our CI in \cref{sec:bias-bound} further requires $s_2$, and our CI in \cref{sec:lincom} requires further strengthening \cref{eqn:theta-hat-dists} to joint normality with known correlation.
Although analytic forms of $s_2$ (or $s_1$) may be difficult to estimate, often bootstrap or subsampling \citep{PolitisRomano1994a} can be used.
We never require $b_2$ because estimating the bias is often infeasible or unreliable.

\subsection{Confidence intervals}
\label{sec:setup-CIs}

With confidence level $1-\alpha$, consider the following two-sided CIs:
\begin{equation}
\label{eqn:CI}
\begin{split}
\CI_1 &\equiv \hat\theta_1 \pm z_{1-\alpha/2}s_1 ,\\
\CI_2 &\equiv \hat\theta_2 \pm z_{1-\alpha/2}s_1 ,\\
\CI_3 &\equiv \hat\theta_1 \pm z_{1-\alpha/2}s_2 ,\\
\CI_4 &\equiv \hat\theta_2 \pm z_{1-\alpha/2}s_2 ,
\end{split}
\end{equation}
where $z_{1-\alpha/2}$ is the $(1-\alpha/2)$-quantile of the standard normal distribution.
The benchmark $\CI_1$ is the usual CI using only the unbiased estimator.
We focus on $\CI_2$, which has not been proposed or studied previously.
We include $\CI_3$ only for completeness; it has not been proposed and is not good.
As noted earlier, $\CI_4$ has been proposed and justified by assuming the bias is asymptotically negligible.

Among the CIs in \cref{eqn:CI}, we propose using $\CI_2$.
To the best of our knowledge, this has not been proposed before, perhaps due to the counterintuitive mixing of one estimator for centering and another estimator's standard error.
Our results characterize the coverage probability advantage of using $\CI_2$ over the benchmark $\CI_1$.

Unlike $\CI_1$ and $\CI_2$, both $\CI_3$ and $\CI_4$ always undercover.
The $\CI_3$ is centered at $\hat\theta_1$ like $\CI_1$ but is shorter because $s_2<s_1$; because $\CI_1$ has exact $1-\alpha$ coverage probability, $\CI_3$ has less than $1-\alpha$ CP.
Also, $\CI_4$ covers $\E(\hat{\theta}_2)$ with probability $1-\alpha$, but $\E(\hat{\theta}_2)=\theta+b_2\ne\theta$, so it has less than $1-\alpha$ CP for $\theta$.

Additionally, in \cref{sec:bias-bound,sec:lincom}, we propose CIs that convert some of the excess coverage probability of $\CI_2$ into shorter length.

\subsection{Applications}
\label{sec:setup-apps}

Our setting is a reasonable approximation for certain MSE-reducing estimators but not others, as we describe below.

Our original research question was simply: is $\CI_2$ valid for smoothed instrumental variables quantile regression \citep{KaplanSun2017,sivqr}?
In that context, $\hat\theta_2$ is generally easier and faster to compute than $\hat\theta_1$, in addition to having lower MSE \citep[][Thm.~7]{KaplanSun2017}, although $s_1$ is also easy to compute \citep[e.g.,][Rmk.~4]{ChernozhukovHansen2006}.
Both $\hat\theta_1$ and $\hat\theta_2$ are asymptotically normal.%
\footnote{%
There are multiple algorithms for computing $\hat\theta_1$, but they all aim to solve for the same unsmoothed solution $\hat\theta_1$ (up to some smaller-order terms) and are consequently all asymptotically normal; for example, see \citet[Rmk.~3]{ChernozhukovHansen2006}, \citet[eqn.~(13)]{ChenLee2018}, and \citet[Thm.~1, Cor.~2]{KaidoWuthrich2021}.
}
Our results suggest that for confidence levels above around $90\%$, $\CI_2$ is not only valid but even better than $\CI_1$.
Further, our $\CI_5$ (\cref{sec:bias-bound}) can provide an even shorter valid CI, and $\CI_6$ (\cref{sec:lincom}) yet shorter.

Our results also apply well to the related setting of smoothed quantile regression (QR).
There too, smoothing the indicator function in the moment conditions (estimating equations) introduces bias but reduces variance, resulting in an overall MSE reduction.
Both the unsmoothed QR estimator and the smoothed QR estimator are asymptotically normal; for example, see \citet[][Thm.~4.2]{KoenkerBassett1978}, \citet[][Thm.~3]{AngristEtAl2006}, \citet[][Thms.~1 and 5(S)]{FernandesEtAl2021}, and \citet[][Thm.~4.3]{HeEtAl2023}.
Asymptotic MSE-optimal smoothing bandwidths are given by \citet[][Prop.~2 and \S5]{KaplanSun2017}, \citet[][Thm.~4]{FernandesEtAl2021}, and \citet[][Rmk.~4.3]{HeEtAl2023}.
Although the difference is generally smaller than with instrumental variables QR, there can be settings (large sample size and/or number of regressors) where the unsmoothed QR estimator is significantly slower than smoothed QR, in which case $\CI_1$ is not convenient; for example, see Figure 1(b) of \citet{HeEtAl2023}.
Our simulations in \cref{sec:sim} illustrate the finite-sample benefits of our proposed CIs in the context of smoothed QR.

With some caution, our results can be applied to other smoothed estimators even when the original unsmoothed estimator is not asymptotically normal.
For example, consider the maximum score estimator \citep{Manski1975} or the nonparametric maximum likelihood estimator in the ``current status'' failure time model \citep{GroeneboomWellner1992}, neither of which is asymptotically normal.
Smoothing attains asymptotic normality in each case.
Further, compared to using a very small smoothing bandwidth that incurs very little bias, increasing the bandwidth increases bias while reducing MSE, up to a point.
In such cases, our $\hat\theta_1$ is a highly under-smoothed (small bandwidth, small bias) estimator, while our $\hat\theta_2$ is the estimator with the (approximate) MSE-optimal bandwidth.
For example, from the results of \citet{Horowitz1992}, Theorem 2(c) provides the asymptotic MSE-optimal bandwidth, which results in the asymptotically normal (but biased) distribution in Theorem 2(b) represented by our $\hat\theta_2$, while an undersmoothed bandwidth sets his $\lambda=0$ and removes the bias in Theorem 2(b), thus serving reasonably as our $\hat\theta_1$.
Similarly, from \citet{GroeneboomEtAl2010}, Theorem 3.5 (or Thm.~3.6 or Cor.~3.7, or Thm.~4.2, Thm.~4.3, or Cor.~4.4) provides the (biased) asymptotically normal distribution with the asymptotic MSE-optimal bandwidth rate for our $\hat\theta_2$, while undersmoothing with bandwidth going to zero faster than $n^{-1/5}$ (or $n^{-1/7}$) reduces the order of bias more and can serve as our $\hat\theta_1$.

With yet more caution, our results can apply to other estimators that reduce MSE in many but not all cases.
For example, using a fixed data-generating process (pointwise asymptotics), the adaptive lasso and SCAD are asymptotically normal and as efficient as the infeasible ``oracle'' estimator knowing the true model (hence lower MSE than the corresponding unpenalized estimator); see Theorems 2 and 4 of \citet{Zou2006} and Theorem 2 of \citet{FanLi2001}.
However, among others, \citet{LeebPotscher2008b} show that such oracle arguments do not hold uniformly (i.e., under all drifting sequences of data-generating processes), which in finite samples translates to portions of the parameter space where MSE is not lower than ordinary least squares; see their Theorem 2.1 and Section 3.
More generally, there may be biased estimators that reduce MSE under certain conditions but not others.
In practice, if such conditions seem plausible enough to use the intentionally biased estimator, then it seems reasonable to construct CIs under the same conditions.
In such cases, while keeping $\hat\theta_1$ and $\CI_1$ as a robustness check, the main results could show $\hat\theta_2$ with one of our proposed CIs.

Finally, our results are suggestive for Stein-like averaging estimators but leave gaps to fill in future work.
Many modern averaging estimators like \citet{Hansen2017} or \citet{ChengLiaoShi2019} take a weighted average of a conservative estimator (like our $\hat\theta_1$) and an aggressive estimator that is biased but usually has lower variance (though not necessarily lower MSE).
Usually the two estimators are jointly asymptotically normal.
The infeasible oracle averaging estimator using the infeasible optimal weight (which is usually a constant) is thus itself asymptotically normal and fits our assumptions about $\hat\theta_2$.
However, unlike in other contexts, the weight cannot be estimated consistently, so it is random even asymptotically.
A random-weighted average of two normals is not normal, so this does not fit our $\hat\theta_2$; for example, see (9) of \citet{Hansen2017} or Lemma 4.2(a) of \citet{ChengLiaoShi2019}.
For weights near zero (or one), the resulting weighted average is nearly normal, and our conclusions may apply more generally if the distribution is ``close enough'' to normal (as seen in simulations like Figure 2 of \citealp{Liu2015} for the plug-in averaging estimator), but formally extending our results to the general case of averaging estimators remains an important gap for future work to fill.

\section{Coverage probability comparison: equal MSE}
\label{sec:equal-MSE}

This section studies the CP of our proposed $\CI_2$ when the biased estimator's MSE equals the unbiased estimator's MSE,
\begin{equation}
\label{eqn:equal-MSE}
\MSE(\hat\theta_2) = b_2^2+s_2^2
= s_1^2 = \MSE(\hat\theta_1) .
\end{equation}

For intuition, consider the extreme case with the maximum possible bias satisfying \cref{eqn:equal-MSE}, $b_2=s_1$ and $s_2=0$.
Note $\CI_3$ and $\CI_4$ both have length zero and thus $0\%$ CP.
Because $s_2=0$, $\hat\theta_2=\theta+b_2=\theta+s_1$ is non-random, so $\CI_2$ simplifies to 
\begin{equation*}
\CI_2=\overbrace{\theta+s_1}^{\hat\theta_2}
\pm z_{1-\alpha/2} s_1 .
\end{equation*}
This $\CI_2$ is also non-random and thus has either $0\%$ or $100\%$ CP.
Specifically, $\CI_2$ includes the true $\theta$ if and only if $s_1-z_{1-\alpha/2}s_1\le0$, which simplifies to $z_{1-\alpha/2}\ge1$ or approximately $1-\alpha\ge0.683$.
Thus, in this special case, $\CI_2$ is better than $\CI_1$ if the confidence level is above $68.3\%$, but worse if below $68.3\%$.
We will more generally characterize when $\CI_2$ is better than $\CI_1$, in terms of both the bias $b_2$ and the confidence level $1-\alpha$.

More generally, to satisfy the equal MSE in \cref{eqn:equal-MSE}, let
\begin{equation}
\label{eqn:sin-cos}
b_2 = s_1 \sin(t)
,\quad
s_2=s_1 \cos(t)
,\quad -\pi/2 \le t \le \pi/2 ,
\end{equation}
where $\sin(\cdot)$ and $\cos(\cdot)$ are the sine and cosine functions.
Intuitively, $t$ is a transformed version of the bias: $b_2$ is increasing in $t$, with $b_2=0$ when $t=0$, up to $b_2=s_1$ when $t=\pi/2$.
The standard deviation $s_2$ moves in the opposite direction, decreasing in $t$ from $s_2=s_1$ at $t=0$ down to $s_2=0$ at $t=\pi/2$.
The earlier special case implicitly set $t=\pi/2$ to get $s_2=0$ and $b_2=s_1$.

The coverage probability of $\CI_2$ depends on $t$ and the critical value $z$:
\begin{align}
\notag
\Pr&(\hat\theta_2-zs_1
     \le\theta\le
     \hat\theta_2+zs_1)
\\&= \notag
\Pr(\hat\theta_2-zs_1\le\theta)
- \Pr(\theta\ge\hat\theta_2+zs_1)
\\&= \notag
\Pr(\hat\theta_2\le\theta+zs_1)
- \Pr(\hat\theta_2\le\theta-zs_1)
\\&= \notag
\Pr\biggl(\overbrace{\frac{\hat\theta_2-\theta-b_2}{s_2}}^{\sim\NormDist(0,1)\textrm{ by \cref{eqn:theta-hat-dists}}}
    \le \frac{\theta+zs_1-\theta-b_2}{s_2}\biggr)
-
\Pr\biggl(\overbrace{\frac{\hat\theta_2-\theta-b_2}{s_2}}^{\sim\NormDist(0,1)\textrm{ by \cref{eqn:theta-hat-dists}}} 
     \le \frac{\theta-zs_1-\theta-b_2}{s_2}\biggr)
\\&= \label{eqn:CI2-CP-b2} 
 \Phi\biggl(\frac{ zs_1-b_2}{s_2}\biggr)
-\Phi\biggl(\frac{-zs_1-b_2}{s_2}\biggr)
\\&= \notag
 \Phi\biggl(\frac{ zs_1-s_1\sin(t)}{s_1\cos(t)}\biggr)
-\Phi\biggl(\frac{-zs_1-s_1\sin(t)}{s_1\cos(t)}\biggr)
\\&= \notag
 \Phi\biggl( \frac{z}{\cos(t)}-\tan(t)\biggr)
-\Phi\biggl(-\frac{z}{\cos(t)}-\tan(t)\biggr)
\\&= \label{eqn:CI2-CP-t}
 \Phi(z\sec(t)-\tan(t))
-\Phi(-z\sec(t)-\tan(t))
\equiv
\CP(t,z) ,
\end{align}
where $\Phi(\cdot)$ is the standard normal CDF, $\tan(\cdot)$ is the tangent function, and $\sec(\cdot)=1/\cos(\cdot)$ is the secant function.
Note $\CP(\cdot,z)$ is an even function because
\begin{align}
\notag
\CP(-t,z)
  &\equiv
 \Phi(z\sec(-t)-\tan(-t))
-\Phi(-z\sec(-t)-\tan(-t))
\\&= \notag
 \Phi(z\sec(t)+\tan(t))
-\Phi(-z\sec(t)+\tan(t))
\\&= \notag
1-\Phi(-z\sec(t)-\tan(t))
-[1-\Phi(z\sec(t)-\tan(t))]
\\&= \notag
 \Phi(z\sec(t)-\tan(t))
-\Phi(-z\sec(t)-\tan(t))
\\&
\equiv \notag
\CP(t,z)
,
\end{align}
using the facts that $\sec(-t)=\sec(t)$, $\tan(-t)=-\tan(t)$, and $\Phi(x)=1-\Phi(-x)$.
Thus, we let $t>0$ without loss of generality.

\begin{figure}[htbp]
\centering
\includegraphics%
{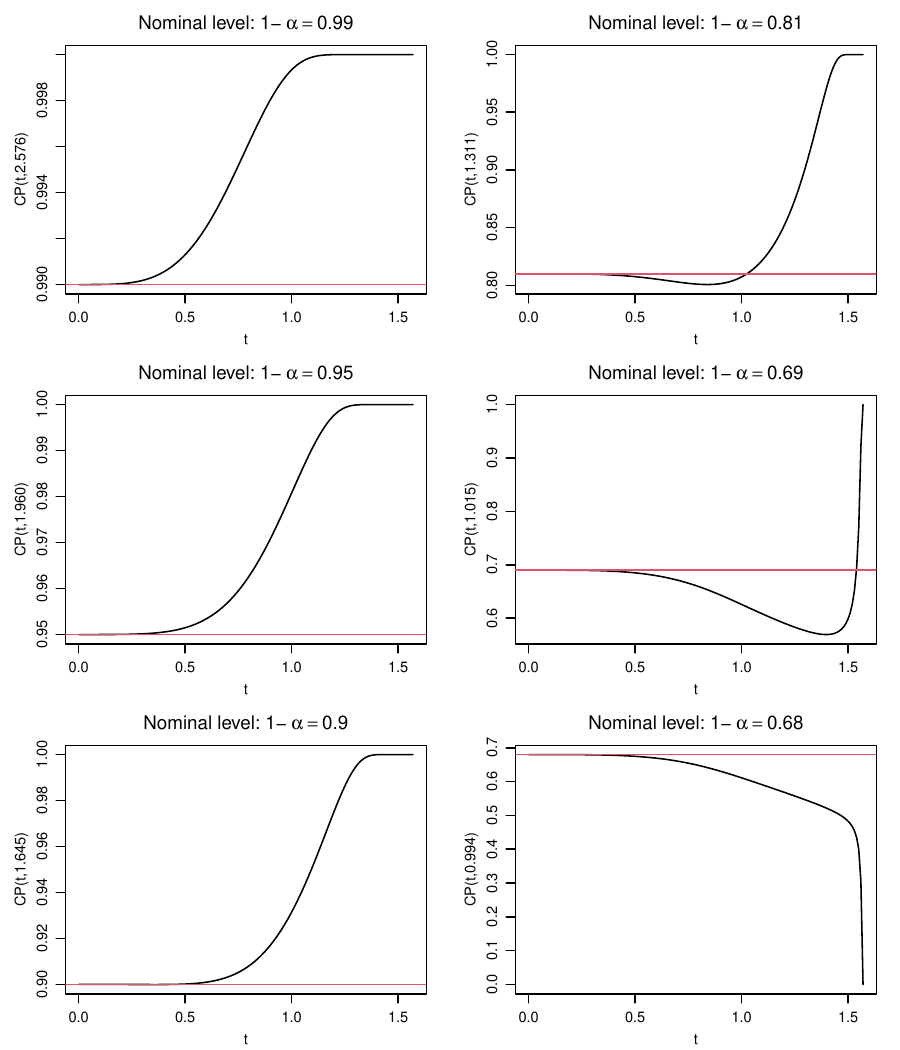}
\caption{\label{fig:cp}%
Coverage probability of $\CI_2$ at various nominal levels.}
\end{figure}

\Cref{fig:cp} shows $\CI_2$'s coverage probability as a function of $t$, for several fixed $z$.
Specifically, each plot shows $\CP(t, z_{1-\alpha/2})$ over $0\le t\le\pi/2$, with the different plots showing confidence levels $1-\alpha=0.99,0.95,0.90,0.81,0.69,0.68$.
The horizontal line in each plot shows $1-\alpha$, which is the coverage probability of the benchmark $\CI_1$, as well as the value of $\CP(0,z_{1-\alpha/2})$ because $t=0$ means $b_2=0$ and $s_2=s_1$, so $\hat\theta_2\deq\hat\theta_1$.
At the usual confidence level $95\%$ (and $99\%$), the coverage probability of $\CI_2$ is always higher than the nominal $1-\alpha$, which is also the CP of $\CI_1$.
At confidence level $90\%$, $\CI_2$ has a worst-case $0.89995$ CP, technically below the nominal $1-\alpha$ but visually and practically indistinguishable; at other $t$, CP exceeds the nominal $1-\alpha=0.9$.
Numerical analysis finds that $1-\alpha\ge0.917$ is the threshold for $\CI_2$ having at least $1-\alpha$ CP uniformly over $t\in[0,\pi/2]$.
At uncommon confidence level $81\%$, the worst-case CP is $0.8007$, and CP exceeds $1-\alpha$ for larger $t$ closer to $\pi/2$.
At the unusual confidence levels $69\%$ and $68\%$, there can be severe undercoverage, especially for $68\%$ because $z_{1-\alpha/2}<1$; this reflects the results from the earlier special case with $t=\pi/2$, in which CP was $100\%$ if $z_{1-\alpha/2}\ge1$ (confidence level above $68.3\%$) but $0\%$ if $z_{1-\alpha/2}<1$ (confidence level below $68.3\%$).
In sum, $\CI_2$ is uniformly better than $\CI_1$ for confidence levels $95\%$ and $99\%$, and practically uniformly better at confidence level $90\%$, but can be much worse with unusually low confidence levels.

\Cref{res:equal-MSE} summarizes these results.

\begin{theorem}
\label{res:equal-MSE}
Consider the setting of \cref{eqn:theta-hat-dists,eqn:CI,eqn:equal-MSE}.
The coverage probability of $\CI_1$ equals the confidence level $1-\alpha$.
For $1-\alpha\ge0.917$, $\CI_2$ has coverage probability strictly above $1-\alpha$ when $b_2\ne0$.
For $1-\alpha=0.9$ and any $b_2$, $\CI_2$ has coverage probability of at least $90.00\%$ (rounded), increasing to $100\%$ as $\abs{b_2}\to s_1$.
\end{theorem}
\begin{proof}
The result for $\CI_1$ is well known.
The coverage probability of $\CI_2$ is derived in \cref{eqn:CI2-CP-t} as $\Phi(z\sec(t)-\tan(t))-\Phi(-z\sec(t)-\tan(t))$, in terms of the parameterization in \cref{eqn:sin-cos}.
Numerical study of this function leads to the stated results, including the worst-case CP with $1-\alpha=0.9$ being $0.899953$, which rounds to $0.9000$.
\end{proof}

\section{Coverage probability comparison: lower MSE}
\label{sec:lower-MSE}

Our finding that $\CI_2$ is better than $\CI_1$ for the most common confidence levels extends from the special case of $\MSE(\hat\theta_2)=\MSE(\hat\theta_1)$ to the general case of $\MSE(\hat\theta_2) \le \MSE(\hat\theta_1)$.
If $\hat\theta_2$ has strictly lower MSE than $\hat\theta_1$, then we can imagine a hypothetical estimator $\hat\theta_B$ ($B$ for ``bound'') with the same variance but higher bias than $\hat\theta_2$, such that $\hat\theta_B$ has the same MSE as $\hat\theta_1$.
As shown in \cref{res:lower-bias-higher-CP}, adding bias like this decreases CP, yet from \cref{res:equal-MSE} we know the CI centered at $\hat\theta_B$ still has higher CP than the benchmark at conventional confidence levels.
Thus, allowing $\hat\theta_2$ to have lower MSE than $\hat\theta_1$ further strengthens the advantages of our proposed $\CI_2$ compared to the benchmark $\CI_1$.

\begin{corollary}
\label{res:lower-MSE}
In \cref{res:equal-MSE}, the same results hold after relaxing the equal-MSE condition in \cref{eqn:equal-MSE} to the weakly-lower-MSE condition in \cref{eqn:lower-MSE}.
\end{corollary}
\begin{proof}
Consider the hypothetical estimator $\hat\theta_B$ ($B$ for ``bound'') with the same variance as $\hat\theta_2$ but larger magnitude bias $b_B$ such that its MSE equals that of the unbiased $\hat\theta_1$: $b_B=\pm\sqrt{s_1^2-s_2^2}$.
From \cref{res:lower-bias-higher-CP}, the CP of $\CI_2$ is at least that of the CI $\hat\theta_B\pm z_{1-\alpha/2}s_1$, to which the results of \cref{res:equal-MSE} apply.
\end{proof}

\begin{lemma}
\label{res:lower-bias-higher-CP}
Let $\hat\theta_B\sim\NormDist(\theta+b_B,s_2^2)$ and $\hat\theta_2\sim\NormDist(\theta+b_2,s_2^2)$ with $\abs{b_2}<\abs{b_B}$.
For any $s_1>0$ and $z>0$, $\CI_2$ ($\hat\theta_2\pm z s_1$) has higher coverage probability than $\hat\theta_B\pm z s_1$:
\begin{equation*}
\Pr(\hat\theta_2 - z s_1
  \le\theta\le \hat\theta_2 + z s_1)
>
\Pr(\hat\theta_B - z s_1
  \le\theta\le \hat\theta_B + z s_1)
.
\end{equation*}
\end{lemma}
\begin{proof}
The CP for each estimator can be written in terms of the standard normal CDF $\Phi(\cdot)$, and then \cref{res:shift-in} can be applied.
First,
\begin{align*}
\Pr(\hat\theta_2 - z s_1
    \le \theta \le
    \hat\theta_2 + z s_1)
  &=
\Pr(\theta - z s_1 \le \hat\theta_2 \le \theta + z s_1)
\\&=
\Pr\biggl(\frac{-z s_1 - b_2}{s_2} \le \overbrace{\frac{\hat\theta_2-\theta-b_2}{s_2}}^{\equiv Z_1\sim\NormDist(0,1)}
   \le \frac{z s_1 - b_2}{s_2} \biggr)
\\&=
 \Phi\biggl(\frac{ z s_1 - b_2}{s_2}\biggr)
-\Phi\biggl(\frac{-z s_1 - b_2}{s_2}\biggr)
,\\
\Pr(\hat\theta_B - z s_1
    \le \theta \le
    \hat\theta_B + z s_1)
  &=
\Pr(\theta - z s_1 \le \hat\theta_B \le \theta + z s_1)
\\&=
\Pr\biggl(\frac{-z s_1 - b_B}{s_2} \le \overbrace{\frac{\hat\theta_B-\theta-b_B}{s_2}}^{\equiv Z_2\sim\NormDist(0,1)}
   \le \frac{z s_1 - b_B}{s_2} \biggr) 
\\&=
 \Phi\biggl(\frac{ z s_1 - b_B}{s_2}\biggr)
-\Phi\biggl(\frac{-z s_1 - b_B}{s_2}\biggr)
.
\end{align*}
The result follows by applying \cref{res:shift-in} with $d\equiv z s_1 / s_2$, $a\equiv -b_2/s_2$, and $b\equiv -b_B/s_2$, noting that the \cref{res:shift-in} condition $\abs{a}<\abs{b}$ is satisfied because $\abs{b_2}<\abs{b_B}$.
\end{proof}

\section{Shorter CI based on bias bound}
\label{sec:bias-bound}

Given the estimator sampling distributions and the MSE bound in \cref{eqn:theta-hat-dists,eqn:lower-MSE}, we can bound the magnitude of the bias.
Then, similar in spirit to the strategy of \citet{ArmstrongKolesar2021}, we construct a symmetric two-sided CI to have exact coverage when the bias attains the upper bound.
By \cref{res:lower-bias-higher-CP}, the coverage is even better with smaller magnitude bias, so this CI has at least $1-\alpha$ coverage probability for all levels of bias satisfying the MSE bound.
Note that this relationship between CP and bias is opposite that of $\CI_2$: here, CP is highest when $b_2=0$, decreasing to $1-\alpha$ as $\abs{b_2}$ increases to the bound, whereas for $\CI_2$ the CP is $1-\alpha$ at $b_2=0$ and (for conventional confidence levels) increases toward $100\%$ as $\abs{b_2}$ increases.

This CI trades some of the higher CP of $\CI_2$ for shorter length.
For example, at confidence level $95\%$ when $\CI_2$ has CP strictly above $95\%$ for any non-zero bias $b_2\ne0$, this alternative CI is strictly shorter than $\CI_2$ while maintaining at least $95\%$ CP.
This is generally preferable to $\CI_2$.
The only additional requirement is the reliance on $s_2$, the biased estimator's standard error.
Usually a consistent estimator is available, corresponding to known $s_2$ in our setting, but if not then $\CI_2$ can still be used.

\subsection{CI definition and properties}

Formally, the new CI is defined as follows.
To facilitate comparison, it is defined the same as $\CI_2$ but replacing $z_{1-\alpha/2}$ with $\tilde{z}_{1-\alpha/2}$, defined implicitly using the $\CP(t,z)$ function from \cref{eqn:CI2-CP-t}:
\begin{equation}
\label{eqn:CI5}
\CI_5 \equiv \hat\theta_2 \pm \tilde{z}_{1-\alpha/2}s_1
,\quad
1-\alpha = \CP\bigl(\cos^{-1}(s_2/s_1), \tilde{z}_{1-\alpha/2}\bigr) .
\end{equation}
This $\tilde{z}_{1-\alpha/2}$ sets CP to $1-\alpha$ when $\abs{b_2}$ satisfies $b_2^2+s_2^2=s_1^2$ (equal MSE), as detailed in the proof of \cref{res:CI5-CP}.

\begin{theorem}
\label{res:CI5-CP}
Given the sampling distributions in \cref{eqn:theta-hat-dists} with known $s_1$ and $s_2$ but unknown $b_2$, for $\CI_5$ defined in \cref{eqn:CI5}, coverage probability is at least $1-\alpha$ for any $b_2$ satisfying the MSE inequality in \cref{eqn:lower-MSE}:
\[ \inf_{-\sqrt{s_1^2-s_2^2} \le b_2 \le \sqrt{s_1^2-s_2^2}} \Pr(\theta\in\CI_5) = 1-\alpha , \]
with the minimum $1-\alpha$ attained when $b_2=\pm\sqrt{s_1^2-s_2^2}$.
\end{theorem}
\begin{proof}
Consider $\CI_5$ with a generic critical value $\tilde{z}$, $\hat\theta_2\pm\tilde{z}s_1$.
Because this has the same structure as $\CI_2$, its CP is the same as derived in \cref{eqn:CI2-CP-b2}, only replacing $z$ with $\tilde{z}$:
\begin{equation}
\label{eqn:CI5-CP}
\Pr(\hat\theta_2-\tilde{z}s_1
    \le \theta
    \le \hat\theta_2+\tilde{z}s_1)
= \Phi(( \tilde{z}s_1-b_2)/s_2)
 -\Phi((-\tilde{z}s_1-b_2)/s_2) .
\end{equation}
This is an even function of $b_2$ because the value does not change when replacing $b_2$ with $-b_2$: using $\Phi(x)=1-\Phi(-x)$,
\begin{align}\notag
&\Phi(( \tilde{z}s_1-(-b_2))/s_2)
-\Phi((-\tilde{z}s_1-(-b_2))/s_2)
\\&= \notag
  1-\Phi((-\tilde{z}s_1-b_2)/s_2)
-[1-\Phi(( \tilde{z}s_1-b_2)/s_2)]
\\&= \label{eqn:CI5-CP-even}
 \Phi(( \tilde{z}s_1-b_2)/s_2)
-\Phi((-\tilde{z}s_1-b_2)/s_2) .
\end{align}
Below, we show that the choice $\tilde{z}=\tilde{z}_{1-\alpha/2}$ defined in \cref{eqn:CI5} sets this CP equal to $1-\alpha$ in the bounding case of equal MSE.
Equal MSE means $b_2^2+s_2^2=s_1^2$, which is the same as using $t=\cos^{-1}(s_2/s_1)$ with $s_2=s_1\cos(t)$ and $b_2=s_1\sin(t)$ as in \cref{eqn:sin-cos}, where using $b_2>0$ is without loss of generality because CP is an even function of $b_2$ as shown in \cref{eqn:CI5-CP-even}.
That is, given $b_2^2+s_2^2=s_1^2$ and $\tilde{z}=\tilde{z}_{1-\alpha/2}$ from \cref{eqn:CI5}, \cref{eqn:CI5-CP} implies
\begin{align*}
\Pr(\theta\in\CI_5)
= \Phi(( \tilde{z}_{1-\alpha/2}s_1-b_2)/s_2)
 -\Phi((-\tilde{z}_{1-\alpha/2}s_1-b_2)/s_2)
= \CP\bigl( \cos^{-1}(s_2/s_1), \tilde{z}_{1-\alpha/2} \bigr) ,
\end{align*}
which by \cref{eqn:CI5} equals $1-\alpha$.
Applying \cref{res:lower-bias-higher-CP} with $b_B=\pm\sqrt{s_1^2-s_2^2}$ and $z=\tilde{z}_{1-\alpha/2}$, this $1-\alpha$ CP is the lower bound when $\abs{b_2}<\abs{b_B}$ given the same $s_1$ and $s_2$.
\end{proof}

In terms of length, $\CI_5$ is shorter than $\CI_2$ essentially whenever $\CI_2$'s CP is above $1-\alpha$, which at the $95\%$ confidence level is always.
That is, compared to $\CI_2$, $\CI_5$ trades some coverage probability for shorter length, which is usually desired as long as CP remains at least $1-\alpha$, which \cref{res:CI5-CP} says it does.
The cases where $\CI_5$ is longer are those in which $\CI_2$ may undercover, which are lower than the conventionally used confidence levels anyway.

\begin{theorem}
\label{res:CI5-length}
Given the sampling distributions in \cref{eqn:theta-hat-dists} with known $s_1$ and $s_2$ (with $s_2<s_1$) but unknown $b_2$ satisfying the MSE inequality in \cref{eqn:lower-MSE}, $\CI_5$ defined in \cref{eqn:CI5} is strictly shorter than $\CI_2$ defined in \cref{eqn:CI} for all confidence levels above $91.7\%$.
For confidence level $90\%$, $\CI_5$ is no more than $0.014\%$ longer than $\CI_2$, and can be significantly shorter.
\end{theorem}
\begin{proof}
Given their similar structures, $\CI_5$ is shorter than $\CI_2$ if and only if $\tilde{z}_{1-\alpha/2}<z_{1-\alpha/2}$.
Consider the implicit definition of $\tilde{z}_{1-\alpha/2}$ in \cref{eqn:CI5}.
Given that the $\CP(\cdot,\cdot)$ function is increasing in its second argument (wider CI has higher CP), if we plug in $z_{1-\alpha/2}$ and get a CP value above $1-\alpha$, then we know $\tilde{z}_{1-\alpha/2}<z_{1-\alpha/2}$.
According to \cref{res:equal-MSE,res:lower-MSE}, for $1-\alpha\ge0.917$, $\CP(t,z_{1-\alpha/2})>1-\alpha$ for all $0<t\le\pi/2$.
Thus, for $1-\alpha\ge0.917$, $\CP\bigl(\cos^{-1}(s_2/s_1),z_{1-\alpha/2}\bigr)>1-\alpha$ for all $s_2<s_1$, implying $\tilde{z}_{1-\alpha/2}<z_{1-\alpha/2}$ and thus $\CI_5$ being shorter than $\CI_2$.

For $1-\alpha=0.9$, numerical analysis find the function $\CP(t,z_{1-\alpha/2})$ is minimized at $t=0.359$ with value $0.899953$; given that worst-case $t$, $0.9=\CP(0.359,1.6451)$, so the solution is $\tilde{z}_{1-\alpha/2}=1.6451$ rather than $z_{1-\alpha/2}=1.6449$, and (without rounding first) $\tilde{z}_{1-\alpha/2}/z_{1-\alpha/2}=1.0001369$, or $0.014\%$ larger.
\end{proof}

\begin{figure}[htbp]
\centering
\includegraphics%
{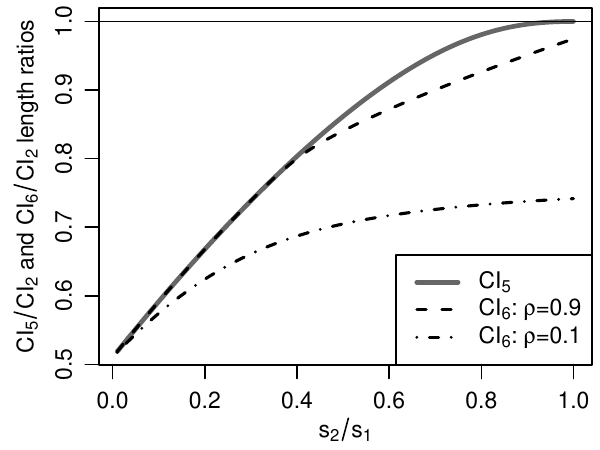}
\caption{\label{fig:length}%
$\CI_5$ (or $\CI_6$) length divided by $\CI_2$ length as function of $s_2/s_1$, $95\%$ confidence level.}
\end{figure}

\Cref{fig:length} shows the length advantage of $\CI_5$ over $\CI_2$ (which has the same length as $\CI_1$) as a function of $s_2/s_1$ at confidence level $1-\alpha=0.95$.
If $s_2$ is near $s_1$, then the lengths are very similar.
If $s_2$ is a small fraction of $s_1$, then $\CI_5$ can be below $60\%$ the length of $\CI_2$.
Even for intermediate values of $s_2/s_1$, $\CI_5$ offers a meaningful advantage in length.
The yet shorter $\CI_6$ is derived in \cref{sec:lincom}.

\subsection{Example}

Consider the properties of $\CI_5$ and $\CI_2$ in the following example with confidence level $95\%$.
Let $s_1=1$ and $s_2=s_1/2$.
From \cref{eqn:CI}, $\CI_2 \equiv \hat\theta_2\pm z_{1-\alpha/2}s_1$, using only $s_1$ and not $s_2$.
From \cref{eqn:CI5}, using both $s_1$ and $s_2$, $\CI_5 \equiv \hat\theta_2\pm \tilde{z}_{1-\alpha/2}s_1$ with $\tilde{z}_{0.975}$ solving $0.95=\CP(\cos^{-1}(s_2/s_1),\tilde{z}_{0.975})$, yielding $\tilde{z}_{0.975}=1.69$.

Regardless of the bias $b_2$, $\CI_5$ is around $14\%$ shorter than $\CI_2$ (which is the same length as $\CI_1$) because $\tilde{z}_{0.975}/z_{0.975}\approx0.86$.

The CP of each CI further depends on the bias $b_2$.
For example, consider $b_2=s_1/2$, so the biased estimator has strictly lower MSE than the unbiased estimator:
\[ \MSE(\hat\theta_2) = b_2^2+s_2^2 = s_1^2/2<s_1^2 = \MSE(\hat\theta_1) . \]
Using \cref{eqn:CI2-CP-b2}, the CP of $\CI_2$ is
\begin{align*}
 \Phi\biggl(\frac{ zs_1-b_2}{s_2}\biggr)
-\Phi\biggl(\frac{-zs_1-b_2}{s_2}\biggr)
  &=
 \Phi\biggl(\frac{ z_{0.975}s_1-s_1/2}{s_1/2}\biggr)
-\Phi\biggl(\frac{-z_{0.975}s_1-s_1/2}{s_1/2}\biggr)
\\&=
 \Phi(2z_{0.975}-1)
-\Phi(-2z_{0.975}-1)
= 0.998 .
\end{align*}
This is better than $\CI_1$ because $\CI_2$ is the same length and has only $0.2\%$ coverage error rate, instead of $5\%$ error rate.
The CP of $\CI_5$ uses the same formula as for $\CI_2$ above but with the lower $\tilde{z}_{0.975}$:
\begin{align*}
 \Phi(2\tilde{z}_{0.975}-1)
-\Phi(-2\tilde{z}_{0.975}-1)
= 0.991 .
\end{align*}
Despite being $14\%$ shorter than $\CI_2$, $\CI_5$ still has CP well above $1-\alpha=0.95$ because $\MSE(\hat\theta_2)$ is well below $\MSE(\hat\theta_1)$.
In this case, it seems well worth trading the $99.8\%$ CP of $\CI_2$ for the slightly lower $99.1\%$ CP of $\CI_5$ in return for a $14\%$ shorter interval.

If instead there is the largest possible bias of $b_2=\sqrt{3}s_1/2$, which implies $\MSE(\hat\theta_2)=\MSE(\hat\theta_1)$, then the CP of both intervals is smaller.
The CP of $\CI_2$ is
\[ \Phi( 2z_{0.975}-\sqrt{3})
  -\Phi(-2z_{0.975}-\sqrt{3})
  = 0.986 ,
 \]
and the CP of $\CI_5$ drops to the nominal $95\%$ level:
\begin{equation*}
 \Phi(2\tilde{z}_{0.975}-\sqrt{3})
-\Phi(-2\tilde{z}_{0.975}-\sqrt{3})
= 0.95 .
\end{equation*}
In that case, $\CI_5$ trades the entire ``excess'' CP of $\CI_2$ in return for the $14\%$ shorter length.

\section{Shortest CI based on convex combination estimator}
\label{sec:lincom}

If we strengthen \cref{eqn:theta-hat-dists} to joint normality with known correlation
\begin{equation}
\label{eqn:joint}
\Biggl(\,
\begin{matrix}
\hat\theta_1\\[-6pt] \hat\theta_2
\end{matrix}
\,\Biggr)
\sim
\NormDist\left(
\Biggl(\,
\begin{matrix}
\theta \\[-6pt] \theta+b_2
\end{matrix}
\,\Biggr)
,\;
\Biggl(\,
\begin{matrix}
s_1^2 & \rho s_1 s_2 \\[-6pt]
\rho s_1 s_2 & s_2^2
\end{matrix}
\,\Biggr)
\right)
,
\end{equation}
then further improvement may be possible by taking a convex combination of $\hat\theta_1$ and $\hat\theta_2$,
\begin{equation}
\label{eqn:theta3}
\hat\theta_{3w} \equiv (1-w)\hat\theta_1 + w\hat\theta_2
,\quad
0\le w\le1 .
\end{equation}
Specifically, our $\CI_6$ defined in \cref{eqn:CI6} is centered at this $\hat\theta_{3w}$, with the known value $w$ chosen by the user to minimize the corresponding CI's length subject to the desired worst-case coverage probability.
As seen below, this optimal $w$ depends on the standard errors $s_1$ and $s_2$ (and $\rho$).
We do not allow more general linear combinations with $w>1$ or $w<0$ because then the MSE of the linear combination may be higher than that of $\hat\theta_1$; similarly, for general linear combination weights that do not sum to $1$, the MSE of the linear combination may be (arbitrarily) higher than that of $\hat\theta_1$ without additional assumptions of bounds on $\theta$.

This setting is similar to Section 4.1 of \citet{ArmstrongKolesar2021}.%
\footnote{Thanks to an anonymous reviewer for making this connection and suggesting we consider convex combinations.}
Our \cref{eqn:joint,eqn:theta3} are similar to their (13) and corresponding linear estimator, except our bound on bias $b_2$ comes from the assumed MSE improvement of $\hat\theta_2$, and our convex combination weight is not required to satisfy any other condition like their (5).

Practically, the strengthening to joint normality is rarely restrictive, but the additional reliance on the correlation $\rho$ may incur extra computation time and estimation error.

The strategy here is to show that the convex combination $\hat\theta_{3w}$ is normally distributed with known variance, so we can apply the same strategy as for $\CI_5$.
One difference is that we know more about the maximum bias, because we know $b_2^2\le s_1^2-s_2^2$ and that the bias of $\hat\theta_{3w}$ is $wb_2$ for known $w$.
This provides a tighter bias bound than simply using $(wb_2)^2\le s_1^2 - s_{3w}^2$ from the MSE inequality.

The new CI is derived as follows.
As formally stated and proved in \cref{res:CP-w-rho}, the CI $\hat\theta_{3w}\pm s_1z$ with general critical value $z$ has coverage probability
\begin{equation}
\label{eqn:CP-w}
\begin{split}
\CP_{w,\rho}(t, z)
&\equiv
 \Phi\biggl(\frac{ z-w\sin(t)}{\sqrt{(1-w)^2+w^2[\cos(t)]^2+2\rho w(1-w)\cos(t)}}\biggr)
\\&\quad
-\Phi\biggl(\frac{-z-w\sin(t)}{\sqrt{(1-w)^2+w^2[\cos(t)]^2+2\rho w(1-w)\cos(t)}}\biggr) ,
\end{split}
\end{equation}
using the parameterization $b_2=s_1\sin(t)$ and $s_2=s_1\cos(t)$ from \cref{eqn:sin-cos}.
Similar to the definition of $\tilde{z}_{1-\alpha/2}$ in \cref{eqn:CI5}, define $\tilde{z}_{w,1-\alpha/2}$ to solve
\begin{equation}
\label{eqn:tilde-z-w}
1-\alpha = \CP_{w,\rho}\bigl( \cos^{-1}(s_2/s_1) , \tilde{z}_{w,1-\alpha/2} \bigr) .
\end{equation}
Define
\begin{equation}
\label{eqn:CI6}
\CI_6 \equiv \hat\theta_{3w^*} \pm s_1 \tilde{z}_{w^*,1-\alpha/2}
,\quad
w^* \equiv \argmin_{0\le w\le1} \tilde{z}_{w,1-\alpha/2} .
\end{equation}
That is, $\tilde{z}_{w,1-\alpha/2}$ is calibrated to achieve the desired coverage probability for any $w$, and $w^*$ is then chosen to minimize the CI's length.

\begin{theorem}
\label{res:CI6}
Given \cref{eqn:lower-MSE,eqn:joint} and the definitions in \cref{eqn:theta3,eqn:CI6}, the coverage probability of $\CI_6$ is at least $1-\alpha$, and the length of $\CI_6$ is less than or equal to the length of $\CI_5$.
\end{theorem}
\begin{proof}
By construction, $\tilde{z}_{w,1-\alpha/2}$ in \cref{eqn:tilde-z-w} sets the coverage probability equal to $1-\alpha$ exactly in the equal-MSE case.
By \cref{res:lower-bias-higher-CP} (replacing $s_2$ with $s_{3w}$, and $b_2$ with $wb_2$, and $b_B=\pm w\sqrt{s_1^2-s_2^2}$), if $\MSE(\hat\theta_2)<\MSE(\hat\theta_1)$ strictly, then the CP is even higher.
The previous $\CI_5$ is the special case with $w=1$; given that $w^*$ minimizes length over $0\le w\le1$, $\CI_6$ is shorter than $\CI_5$ when $w^*<1$ and has equal length when $w^*=1$.
\end{proof}

\Cref{fig:length} (from earlier) shows how the benefit of $\CI_6$ depends on the correlation $\rho=\Corr(\hat\theta_1,\hat\theta_2)$ and the ratio $s_2/s_1$.
If $\rho$ is near $1$, then $\hat\theta_1$ contains little information not already contained in $\hat\theta_2$, so taking a convex combination cannot improve much over simply using $\hat\theta_2$.
If $s_2/s_1$ is small, then $\CI_5$ already has much shorter length than $\CI_2$, so the CI-optimal convex combination will put most or all of the weight on $\hat\theta_2$ anyway; that is $w^*\approx1$, so $\CI_6\approx\CI_5$.
Conversely, if $s_2/s_1$ is closer to one and the correlation $\rho$ is farther below one, then using $\hat\theta_{3w}$ can shorten the CI length considerably.
For example, with $\rho=0.1$ and $s_2/s_1=1$, the length of $\CI_6$ is only $74\%$ of the length of $\CI_5$.

\section{Simulation}
\label{sec:sim}

The following simulation illustrates the finite-sample properties of our proposed CIs in the context of smoothed quantile regression (SQR).
We use the R code for the smoothed instrumental variables quantile regression estimator of \citet{KaplanSun2017}, which inclues SQR as a special case (when the instrument vector equals the regressor vector), using their automatic plug-in version of the bandwidth formula in their Proposition 2 whose rate minimizes the asymptotic mean squared error (see also their Section 5 including Theorem 7).
Additional theoretical results for SQR are provided by \citet{FernandesEtAl2021}, including allowing a stochastic bandwidth sequence (Theorem 5(S)) and uniformity over ranges of the quantile index, and \citet{HeEtAl2023} further extend SQR to a growing number of regressors.
Code in R \citep{R.core} to replicate our simulation results is available on the first author's website.%
\footnote{\url{https://kaplandm.github.io/}}

Within each simulation replication, we do the following.
First, data are generated by
\begin{equation}
\label{eqn:DGP}
Y_i = \beta_0 + \theta X_i + U_i
,\quad
X_i\iid\NormDist(0,1)
,\quad
U_i\iid\NormDist(0,1)
,\quad X_i\independent U_i
,\quad i=1,\ldots,n ,
\end{equation}
with $\beta_0=1$ and $\theta=2$.
In this case, the slope of the conditional $\tau$-quantile function equals $\theta$ for all $0<\tau<1$.
Second, we use nonparametric pairs bootstrap to estimate the standard errors $s_1$ and $s_2$ from \cref{eqn:theta-hat-dists} and the correlation $\rho$ from \cref{eqn:joint}.
Within each bootstrap replication, we compute estimator $\hat\theta_1$ using \code{rq()} from the \code{quantreg} package \citep{R.quantreg} and compute estimator $\hat\theta_2$ using the code based on \citet{KaplanSun2017}; then, using all bootstrapped estimates, we compute the standard deviations and correlation.
Third, with $1-\alpha=0.95$, we compute $\CI_1$ and $\CI_2$ from \cref{eqn:CI}, $\CI_5$ from \cref{eqn:CI5}, and $\CI_6$ from \cref{eqn:CI6}.
As a crude alternative to help avoid undercoverage due to estimation error in $\rho$, we compute $\CI_6^s$ using value $(1+\rho)/2$.

\begin{table}[htbp]
\iftoggle{REV}{\bfseries}{}
\centering
\caption{\label{tab:sim}\iftoggle{REV}{\bfseries}{}Simulated coverage probability and median length, $1-\alpha=0.95$.}
\sisetup{round-precision=3, detect-weight, mode=text}
\begin{threeparttable}
\begin{tabular}[c]{S[table-format=4.0,round-precision=0] S[table-format=1.2,round-precision=2] 
S[table-format=1.3,round-precision=3]
S[table-format=1.3,round-precision=3]
S[table-format=1.3,round-precision=3]
S[table-format=1.3,round-precision=3]
S[table-format=1.3,round-precision=3]
c
S[table-format=1.3,round-precision=3]
S[table-format=1.3,round-precision=3]
S[table-format=1.3,round-precision=3]
S[table-format=1.3,round-precision=3]
S[table-format=1.3,round-precision=3] }
\toprule
 & & \multicolumn{5}{c}{Coverage probability (CP)}
 & & \multicolumn{5}{c}{Median length} \\
\cmidrule{3-7}
\cmidrule{9-13}
{$n$} & {$\tau$} &
{$\CI_1$} & {$\CI_2$} & {$\CI_5$} & {$\CI_6^s$} & {$\CI_6$} &&
{$\CI_1$} & {$\CI_2$} & {$\CI_5$} & {$\CI_6^s$} & {$\CI_6$} \\
\midrule
 100 & 0.10 & 0.9540 & 0.9780 & 0.9760 & 0.9660 & 0.9620 &&   0.7131 &   0.7131 &   0.7038 &   0.6830 &   0.6617 \\
 100 & 0.20 & 0.9720 & 0.9820 & 0.9760 & 0.9720 & 0.9660 &&   0.5861 &   0.5861 &   0.5682 &   0.5508 &   0.5310 \\
 100 & 0.30 & 0.9640 & 0.9800 & 0.9740 & 0.9700 & 0.9680 &&   0.5413 &   0.5413 &   0.5296 &   0.5147 &   0.5016 \\
 100 & 0.40 & 0.9460 & 0.9840 & 0.9740 & 0.9620 & 0.9500 &&   0.5156 &   0.5156 &   0.5040 &   0.4914 &   0.4750 \\
 100 & 0.50 & 0.9600 & 0.9900 & 0.9880 & 0.9720 & 0.9700 &&   0.5138 &   0.5138 &   0.4965 &   0.4829 &   0.4672 \\[2pt]
 200 & 0.10 & 0.9460 & 0.9700 & 0.9640 & 0.9540 & 0.9460 &&   0.5063 &   0.5063 &   0.4912 &   0.4757 &   0.4610 \\
 200 & 0.20 & 0.9460 & 0.9700 & 0.9640 & 0.9540 & 0.9500 &&   0.4043 &   0.4043 &   0.3958 &   0.3854 &   0.3745 \\
 200 & 0.30 & 0.9460 & 0.9700 & 0.9660 & 0.9600 & 0.9480 &&   0.3728 &   0.3728 &   0.3668 &   0.3570 &   0.3477 \\
 200 & 0.40 & 0.9360 & 0.9680 & 0.9600 & 0.9460 & 0.9420 &&   0.3641 &   0.3641 &   0.3559 &   0.3467 &   0.3367 \\
 200 & 0.50 & 0.9580 & 0.9760 & 0.9700 & 0.9680 & 0.9600 &&   0.3583 &   0.3583 &   0.3497 &   0.3392 &   0.3288 \\[2pt]
1000 & 0.10 & 0.9420 & 0.9680 & 0.9620 & 0.9480 & 0.9440 &&   0.2158 &   0.2158 &   0.2128 &   0.2076 &   0.2026 \\
1000 & 0.20 & 0.9420 & 0.9620 & 0.9600 & 0.9480 & 0.9420 &&   0.1772 &   0.1772 &   0.1753 &   0.1718 &   0.1684 \\
1000 & 0.30 & 0.9500 & 0.9640 & 0.9580 & 0.9540 & 0.9460 &&   0.1662 &   0.1662 &   0.1640 &   0.1597 &   0.1566 \\
1000 & 0.40 & 0.9460 & 0.9640 & 0.9600 & 0.9560 & 0.9500 &&   0.1584 &   0.1584 &   0.1568 &   0.1532 &   0.1497 \\
1000 & 0.50 & 0.9580 & 0.9780 & 0.9720 & 0.9660 & 0.9660 &&   0.1568 &   0.1568 &   0.1549 &   0.1512 &   0.1475 \\
\bottomrule
\end{tabular}
\begin{tablenotes}[para,flushleft]
\footnotesize{}
Nominal confidence level $1-\alpha=0.95$, 500 simulation replications, 399 bootstrap replications.
\end{tablenotes}
\end{threeparttable}
\end{table}

\Cref{tab:sim} shows the simulated coverage probability and median length of each CI, with the following patterns.
First, as expected, among our four proposed CIs, both CP and median length decrease from left to right: $\CI_2$ is highest (for both CP and length), then $\CI_5$, then $\CI_6^s$, and finally $\CI_6$ has lowest CP and length.
That is, the CIs progressively trade more of $\CI_2$'s ``extra'' CP (above $1-\alpha$) for shorter length.
Second, compared to $\CI_1$, $\CI_2$ has the same length (by construction) and higher CP.
That is, $\CI_2$ makes fewer coverage errors than $\CI_1$, without sacrificing length.
Third, compared to $\CI_1$, $\CI_5$ always has higher CP as well as always having shorter length.
That is, it successfully trades some of $\CI_2$'s extra CP for shorter length.
Fourth, compared to $\CI_1$, $\CI_6^s$ always has weakly higher CP while (by construction) reducing length even more than $\CI_5$.
Fifth, compared to $\CI_1$, $\CI_6$ has weakly higher CP in all but two cases while always achieving the shortest length among all the CIs.
Despite having lower CP than $\CI_1$ in two cases (and undercovering in one of those), $\CI_6$ fixes $\CI_1$'s undercoverage in three other cases and increases CP at least $0.5$ percentage points in yet four other cases.
That is, although $\CI_6$ does not achieve uniformly better CP than the benchmark $\CI_1$, arguably it still improves CP overall, even while reducing length by 5--10\%.

\section{Empirical illustration}
\label{sec:emp}

Like the simulation, our empirical illustration also uses smoothed quantile regression, also with replication code on the first author's website.%
\footnote{\url{https://kaplandm.github.io/}}
We use the 1980 and 1990 Census data from \citet{AngristEtAl2006}, and like their Figure 2A we focus on the coefficient on years of schooling in a log earnings regression that also includes experience and its square; see Section 4 of \citet{AngristEtAl2006} for details about the data.
They use data for U.S.-born men aged 40--49, and we further restrict to Black men.
Like them, we multiply the schooling coefficient by $100$ to get the approximate return to schooling as a percent.
As in the simulation, we use confidence level $1-\alpha=0.95$ and $399$ bootstrap replications.

\begin{table}[htbp]
\centering
\caption{\label{tab:emp}
Confidence intervals for schooling coefficient.}
\sisetup{round-precision=2, table-format=2.2, detect-weight, mode=text}
\begin{threeparttable}
\begin{tabular}[c]{c S[round-precision=2,table-format=1.2]
    >{{[}} 
    S[table-space-text-pre={[}]
    @{,\,} 
    S[table-space-text-post={]}]
    <{{]}} 
    >{{[}} 
    S[table-space-text-pre={[}]
    @{,\,} 
    S[table-space-text-post={]}]
    <{{]}} 
    >{{[}} 
    S[table-space-text-pre={[}]
    @{,\,} 
    S[table-space-text-post={]}]
    <{{]}} 
    S[round-precision=1,table-format=2.1]}
\toprule
&&\multicolumn{2}{c}{}&\multicolumn{2}{c}{}&\multicolumn{2}{c}{}&{$\CI_6/\CI_1$} \\
{Census} & {$\tau$} 
& \multicolumn{2}{c}{$\CI_1$} & \multicolumn{2}{c}{$\CI_5$}
& \multicolumn{2}{c}{$\CI_6$} & {length (\%)} \\
\midrule
1980 & 0.10 &  5.82503 &  9.76257 &  6.20824 & 10.13942 &  6.08002 &  9.88979 & 96.75502 \\
     & 0.50 &  6.64533 &  8.27107 &  6.43086 &  8.04982 &  6.58293 &  8.11378 & 94.16354 \\
     & 0.90 &  4.06973 &  6.11499 &  4.12323 &  6.15116 &  4.15641 &  6.07384 & 93.75004 \\[2pt]
1990 & 0.10 & 12.11619 & 15.96512 & 12.52554 & 16.31641 & 12.43531 & 16.03651 & 93.56360 \\
     & 0.50 & 11.26703 & 12.85505 & 11.22699 & 12.81176 & 11.27685 & 12.80331 & 96.12303 \\
     & 0.90 &  8.31466 & 10.57239 &  8.47758 & 10.72162 &  8.45624 & 10.58884 & 94.45783 \\
\bottomrule
\end{tabular}
\begin{tablenotes}[para,flushleft]
\footnotesize{}
Nominal confidence level $1-\alpha=0.95$, 399 bootstrap replications.
Following Figure 2A of \citet{AngristEtAl2006}, schooling coefficient values are multiplied by $100$.
Samples are U.S.-born Black men aged 40--49 in the Census year shown).
\end{tablenotes}
\end{threeparttable}
\end{table}

\Cref{tab:emp} reports $\CI_1$, $\CI_5$, and $\CI_6$, for both Census years and quantile levels $\tau\in\{0.1,0.5,0.9\}$.
Briefly, although not relevant to our methodology, we note some similarities and differences with Figure 2A of \citet{AngristEtAl2006}, whose sample includes Black men but is dominated by white men.
Both our results show the same upward shift in schooling coefficients from 1980 to 1990, but their results show the $\tau=0.9$ schooling coefficient increasing much faster than (and in 1990 well exceeding) the $\tau\in\{0.1,0.5\}$ coefficients, whereas our results for Black men show a consistently \emph{lower} return to schooling at the $0.9$-quantile level than at lower quantiles.

\Cref{tab:emp} also reports the length of our $\CI_6$ as a percent of the benchmark $\CI_1$, with modest but consistent improvements of $3$--$7$\% that are similar to the $5$--$10$\% improvements from the simulation.
Unfortunately, as usual, we cannot know the true coverage in empirical examples, so we cannot see the other benefit of increased coverage probability shown in the theory and simulations, where $\CI_5$ always had higher coverage probability than $\CI_1$, and $\CI_6$ almost always.

\section{Conclusion}

In practice, given an intentionally biased estimator that has at least as good mean squared error as the corresponding unbiased estimator, as a default choice we recommend $\CI_5$, our CI centered at the biased estimator using the unbiased estimator's standard error and with critical value calibrated to be shorter than a CI using the unbiased estimator's standard error.
This CI is easy to compute, avoids the undercoverage from using the biased estimator's standard error, and for conventional confidence levels is better than the benchmark of using on the unbiased estimator and its standard error, having both shorter length and higher coverage probability.
In some cases our $\CI_6$ is even shorter by using the correlation to solve for the convex combination estimator that minimizes CI length while maintaining coverage.
Our results are not model-specific, so they apply not only to all the estimators cited in the introduction but also any future estimators that introduce bias in order to reduce mean squared error.
If the biased estimator's standard error is not reliably estimated, then using the usual critical value (and the unbiased estimator's standard error) also improves upon other confidence intervals and is even easier to compute in practice.

In future work, it would be valuable to extend these results more precisely to non-normal distributions, specifically those arising in averaging estimators with random weights like in \citet{Hansen2017} and \citet{ChengLiaoShi2019}.
It would also be interesting to consider confidence sets for vector-valued parameters, as in papers going back to \citet{Stein1962}, as well as uniform confidence bands based on function-valued estimators (such as nonparametric regression or smoothed CDF estimators) with asymptotic Gaussian process limits: under what conditions on the bias function $b(\cdot)$ and covariance functions $s_1(\cdot,\cdot)$ and $s_2(\cdot,\cdot)$ would a uniform confidence band have higher coverage probability when centered at the (more) biased estimator?

\section*{Acknowledgments}

We thank Editor Esfandiar Maasoumi and the anonymous associate editor and referees for their help improving this paper, as well as Alyssa Carlson for feedback on multiple drafts.

\singlespacing
\bibliographystyle{jpe}


\appendix
\paperspacing

\section{Additional results and proofs}
\label{sec:app-proofs}

% ...

% ...

% ...

% ...

% ...

% ...

\begin{lemma}
\label{res:shift-in}
Let $a,b,d\in\R$ with $\abs{a}<\abs{b}$ and $d>0$.
Then, with $\Phi(\cdot)$ the standard normal CDF, $\Phi(a+d)-\Phi(a-d)>\Phi(b+d)-\Phi(b-d)$.
\end{lemma}
\begin{proof}
Let $0\le a<b$ without loss of generality because of the symmetry of the standard normal distribution.
That is, if the original $a<0$, we can replace it with $-a$ without changing the probability: using $\Phi(-z)=1-\Phi(z)$,
\begin{align*}
\Phi(a+d)-\Phi(a-d)
  &= [1-\Phi(-(a+d))]-[1-\Phi(-(a-d))]
\\&= \Phi(-a+d)-\Phi(-a-d)
,
\end{align*}
and similarly we could replace any $b<0$ with $-b$.

We want to show the following is positive:
\begin{align*}
&[\Phi(a+d)-\Phi(a-d)]
-[\Phi(b+d)-\Phi(b-d)]
\\&=
 [\Phi(b-d)-\Phi(a-d)]
-[\Phi(b+d)-\Phi(a+d)]
\\&=
 \int_{a-d}^{b-d}\phi(x)\diff{x}
-\int_{a+d}^{b+d}\phi(x)\diff{x}
\\&=
 \int_0^{b-a}\phi(x+a-d)\diff{x}
-\int_0^{b-a}\phi(x+a+d)\diff{x}
\\&=
\int_0^{b-a}\overbrace{[\phi(x+a-d)-\phi(x+a+d)]}^{>0\textrm{ by \cref{eqn:phi-ineq}}}\diff{x}
> 0 .
\end{align*}
Because $\phi(z)$ is decreasing in $\abs{z}$, and $a,x\ge0$ and $d>0$, then $\abs{x+a-d}<\abs{x+a+d}$, so
\begin{equation}
\label{eqn:phi-ineq}
\phi(x+a-d)>\phi(x+a+d),
\end{equation}
making the above expression positive as desired.
\end{proof}

\begin{lemma}
\label{res:CP-w-rho}
Given \cref{eqn:joint,eqn:theta3}, given any $0\le w\le1$ and any $-1\le\rho\le1$, in terms of critical value $z$ and the value $t$ such that $s_2=s_1\cos(t)$ and $b_2=s_1\sin(t)$, the coverage probability of $\hat\theta_{3w}\pm s_1z$ for $\theta$ is given by \cref{eqn:CP-w}.
\end{lemma}
\begin{proof}
Given \cref{eqn:joint,eqn:theta3},
\begin{equation}
\label{eqn:theta3-hat-dist}
\hat\theta_{3w} \sim
\NormDist\bigl( \theta+wb_2, s_{3w}^2 \bigr)
,\quad
s_{3w}^2 \equiv
(1-w)^2s_1^2+w^2s_2^2+2\rho w(1-w)s_1s_2
.
\end{equation}
This parallels the normal distribution of $\hat\theta_2$ in \cref{eqn:theta-hat-dists}, but with bias $wb_2$ (instead of $b_2$) and standard deviation $s_{3w}$ (instead of $s_2$).
Thus, for a given critical value $z$ and $w$, the coverage probability of $\hat\theta_{3w}\pm s_1z$ has the same structure as \cref{eqn:CI2-CP-b2} but now with bias $wb_2$ (instead of $b_2$) and denominator $s_{3w}$ (instead of $s_2$):
\begin{equation*}
\Pr(\hat\theta_{3w} - s_1z \le \theta \le \hat\theta_{3w} + s_1z)
=\Phi\biggl(\frac{ zs_1-wb_2}{s_{3w}}\biggr)
-\Phi\biggl(\frac{-zs_1-wb_2}{s_{3w}}\biggr) .
\end{equation*}
Plugging in for $s_{3w}$ from \cref{eqn:theta3-hat-dist} and using the same equal-MSE parameterization from \cref{eqn:sin-cos} of $b_2=s_1\sin(t)$ and $s_2=s_1\cos(t)$,
\begin{align*}
 \Phi\biggl(\frac{ zs_1-wb_2}{s_{3w}}\biggr)
-\Phi\biggl(\frac{-zs_1-wb_2}{s_{3w}}\biggr)
  &=
 \Phi\biggl(\frac{ zs_1-ws_1\sin(t)}{s_{3w}}\biggr)
-\Phi\biggl(\frac{-zs_1-ws_1\sin(t)}{s_{3w}}\biggr)
\\&=
 \Phi\biggl(\frac{ z-w\sin(t)}{s_{3w}/s_1}\biggr)
-\Phi\biggl(\frac{-z-w\sin(t)}{s_{3w}/s_1}\biggr) ,
\end{align*}
where, noting $s_2/s_1=\cos(t)$,
\begin{align*}
s_{3w}/s_1
  &=
\sqrt{[(1-w)^2s_1^2+w^2s_2^2+2\rho w(1-w)s_1s_2]/s_1^2}
\\&=
\sqrt{(1-w)^2+w^2s_2^2/s_1^2+2\rho w(1-w)s_2/s_1}
\\&=
\sqrt{(1-w)^2+w^2[\cos(t)]^2+2\rho w(1-w)\cos(t)} .
\end{align*}
The expression in \cref{eqn:CP-w} follows.
\end{proof}

\end{document}